\documentclass[letterpaper, 10 pt, conference]{ieeeconf}  

\IEEEoverridecommandlockouts                              

\usepackage{graphics} 
\usepackage{epsfig} 
\usepackage{mathptmx} 
\usepackage{times} 
\usepackage{amsmath} 

\usepackage{amsthm}
\usepackage{amssymb}  
\usepackage{wasysym}
\usepackage{txfonts}
\usepackage{bbm}
\usepackage[singlelinecheck=false,font=footnotesize]{caption}
\usepackage{subcaption}
\usepackage{color}  
\usepackage[dvipsnames]{xcolor}
\usepackage{indentfirst}
\usepackage{multirow}

\usepackage{color, colortbl}

\usepackage{algorithm}
\usepackage{algpseudocode}
\usepackage[
    style=ieee,
    doi=false,
    isbn=false,
    url=false,
    eprint=false,
    backend=bibtex,
    natbib=true
    ]{biblatex}
    
\usepackage{hyperref}    

\graphicspath{{Images/}}

\bibliography{references}

\newtheorem{defn}{Definition}
\newtheorem{rem}[defn]{Remark}
\newtheorem{lem}[defn]{Lemma}

\newtheorem{assum}[defn]{Assumption}

\providecommand{\R}{\ensuremath \mathbb{R}}
\providecommand{\N}{\ensuremath \mathbb{N}}

\providecommand{\X}{\ensuremath \mathcal{X}}
\renewcommand{\P}{\ensuremath \mathcal{P}}

\newcommand{\norm}[1]{\left\Vert#1\right\Vert}

\newcommand{\vep}{\varepsilon}

\newcommand{\inv}{^{-1}}

\newcommand{\bd}[1]{\partial #1}

\providecommand{\frs}{_\mathrm{FRS}}
\providecommand{\obs}{_\mathrm{obs}}

\providecommand{\stp}{_\mathrm{stop}}
\providecommand{\safe}{_\mathrm{safe}}

\newcommand{\ts}[1]{\textsuperscript{#1}}

\newcommand{\regtext}[1]{\mathrm{\textnormal{#1}}}

\newcommand{\hi}{_\regtext{hi}}
\newcommand{\hio}{_{\regtext{hi},0}}
\newcommand{\hii}{_{\regtext{hi},i}}

\newcommand{\plan}{_\regtext{plan}}
\newcommand{\sense}{_\regtext{sense}}

\newcommand{\des}{_\regtext{des}}

\newcommand{\vmax}{{v_\mathrm{max}}}

\newcommand{\al}{\alpha}

\definecolor{Gray}{gray}{0.9}
\newcolumntype{g}{>{\columncolor{Gray}}c}

\title{\LARGE \bf
Guaranteed Safe Reachability-based Trajectory Design for a High-Fidelity Model of an Autonomous Passenger Vehicle}
\author{Sean Vaskov\ts{1}, Utkarsh Sharma\ts{2}, Shreyas Kousik\ts{1},\\ Matthew Johnson-Roberson\ts{3}, Ramanarayan Vasudevan\ts{1}
\thanks{* This work is supported by the Ford Motor Company via the Ford-UM Alliance under award N022977.}
\thanks{$^{1}$Mechanical Engineering, University of Michigan, Ann Arbor, MI 48109
 {\tt\small <skvaskov,skousik,ramv>@umich.edu}}%
\thanks{$^{2}$Integrative Systems + Design, University of Michigan, Ann Arbor, MI 48109
        {\tt\small <utkrsh>@umich.edu}}%
\thanks{$^{3}$Naval Architecture and Marine Engineering, University of Michigan, Ann Arbor, MI 48109
        {\tt\small <mattjr>@umich.edu}}%
}
\begin{document}

\maketitle
\thispagestyle{empty}

\begin{abstract}
Trajectory planning is challenging for autonomous cars since they operate in unpredictable environments with limited sensor horizons.
To incorporate new information as it is sensed, planning is done in a loop, with the next plan being computed as the previous plan is executed.
The recent Reachability-based Trajectory Design (RTD) is a provably safe, real-time algorithm for trajectory planning.
RTD consists of an offline component, where a Forward Reachable Set (FRS) is computed for the vehicle tracking parameterized trajectories; and an online part, where the FRS is used to map obstacles to constraints for trajectory optimization in a provably-safe way.
In the literature, RTD has only been applied to small mobile robots.
The contribution of this work is applying RTD to a passenger vehicle in CarSim, with a full powertrain model, chassis and tire dynamics.
RTD produces safe trajectory plans with the vehicle traveling up to 15 m/s on a two-lane road, with randomly-placed obstacles only known to the vehicle when detected within its sensor horizon.
RTD is compared with a Nonlinear Model-Predictive Control (NMPC) and a Rapidly-exploring Random Tree (RRT) approach.
The experiment demonstrates RTD's ability to plan safe trajectories in real time, in contrast to the existing state-of-the-art approaches.
\end{abstract}
\section{Introduction}\label{sec:introduction}

\begin{figure}[t]
\centering
    \begin{subfigure}[t]{0.5\textwidth}
        \centering
        \includegraphics[width=0.9\columnwidth]{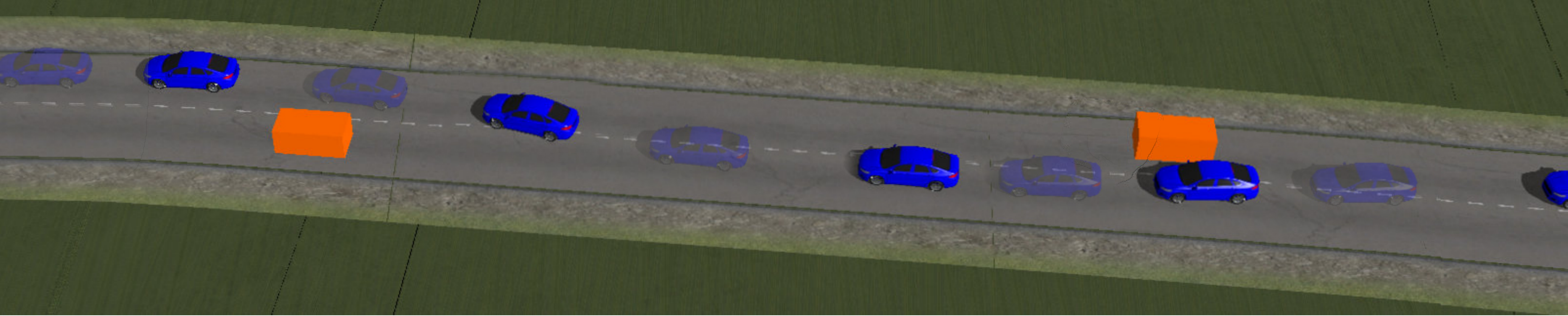}
        \caption{\centering}
        \label{subfig:frs_for_carsim_fusion}
    \end{subfigure}
    
    \begin{subfigure}[t]{0.5\textwidth}
        \centering
        \includegraphics[width=0.9\columnwidth]{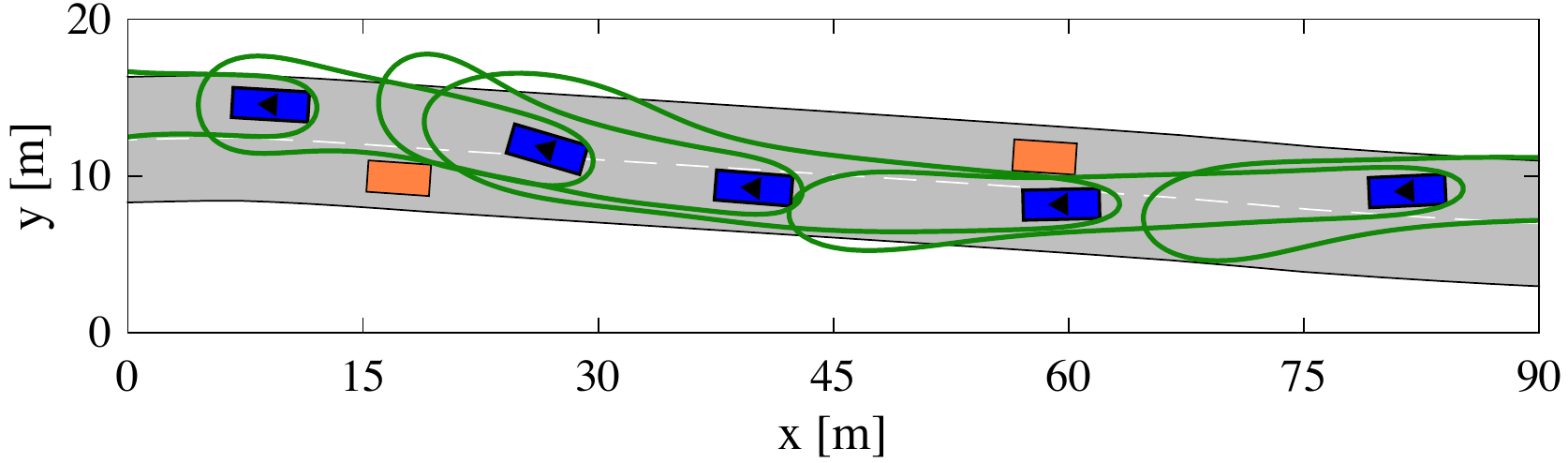}
        \caption{\centering}
        \label{subfig:carsim_fusion}
    \end{subfigure}
    \caption{    Simulation of the vehicle performing two lane change maneuvers on a $90$ m section of a $1$ km test track,beginning from the right side of the figure.
    It uses the presented RTD method to drive safely while avoiding randomly generated obstacles in real time.
    Subfigure (a) depicts the vehicle in Carsim; (b) depicts the RTD planner, which generates trajectories autonomously using an optimization toolbox in MATLAB. 
    Orange boxes are obstacles. 
    The vehicle is solid blue at the same time instances in both plots; in (a), the vehicle is transparent at intermediate times to illustrate its motion.
    The green contours in (b) represent the forward reachable set corresponding to the trajectory at each planning iteration.
    The vehicle begins the first lane change at approximately $15$ m/s, then slows down to approximately $8$ m/s while completing the second.
    A video is available at \protect\url{http://www.roahmlab.com/acc2019_rtd_video}.}
    \label{fig:carsim_and_matlab_intro}
\end{figure}

Autonomous vehicles typically operate with a limited sensor horizon in unpredictable environments.
To do so, they often employ a three-level hierarchy for receding-horizon motion planning, wherein a short trajectory is executed while the next trajectory is being planned \cite{buehler2009darpa,gray2012predictive,theta_star_rrt}.
In this hierarchy, the high-level planner provides coarse route information without considering vehicle dynamics.
The mid-level planner, or trajectory planner, creates dynamically-feasible trajectories and associated control inputs that guide the vehicle along the high-level planner's route while avoiding obstacles.
The low-level controller translates the mid-level planner's control inputs into commands for the vehicle's actuators, without considering the vehicle's surroundings.
To ensure safety, the trajectory planner must compensate for uncertainty in a vehicle's model, which can appear as state estimation error and tracking error between planned and executed trajectories.
Furthermore, the trajectory planner must be \emph{persistently feasible}, meaning that it is always able to find a new, safe trajectory while executing the previous trajectory.
This means that the trajectory planner must produce safe, dynamically-feasible trajectories in real time.
This is challenging because the dynamics of a vehicle are typically nonlinear and high-dimensional \cite{carsim, orosznonholonomic, kuwata2009rrt}.
In this work, we apply a recent proposed trajectory planner that is provably safe and persistently feasible, called Reachability-based Trajectory Design (RTD) \cite{RTD_big_paper}, to a passenger vehicle as depicted in Figure \ref{fig:carsim_and_matlab_intro}.
In this work, we only consider static obstacles; dynamic obstacles are left as future work.

\subsection{Literature Review}\label{subsec:lit_review}
A variety of approaches have been proposed to attempt safe, persistently-feasible trajectory planning.
We briefly review existing trajectory planning methods, which can be broadly divided into sampling-based, model predictive control, and reachability-based approaches.

Sampling-based approaches operate by discretizing the vehicle's state or control space and time to find trajectories that follow high-level routes \cite{elbanhawi2014sampling}.
For example, the Rapidly-exploring Random Tree (RRT) and Probabilistic RoadMaps (PRM) algorithms attempt to find dynamically-feasible trajectories by growing graphs of nodes in the vehicle's state space, with edges between the nodes associated with control inputs; since vehicles have complex dynamic models, new nodes for such systems are created by forward-integration of the dynamics \cite{karaman2011sampling}.
However, it is difficult to guarantee safety of these approaches for two reasons.
First, they typically only track a vehicle's center of mass (as opposed to the vehicle's body), which can make collision checking for a vehicle's entire body challenging; and second, for complex nonlinear models, one must typically specify a temporal and state space discretization granularity for collision checking \cite{elbanhawi2014sampling}.
Importantly, the finer the discretization, the slower a collision checker runs; i.e. there is a tradeoff between speed and safety.
To encourage safety and persistent feasibility, one can plan a braking trajectory at the same time as a non-braking trajectory (i.e. one that attempts to satisfy the high-level plan) \cite{kuwata2009rrt}; or use an Extended Kalman Filter when propagating an RRT to compensate for a vehicle's inability to perfectly follow a planned trajectory \cite{pepy2006rrt}.
However, to the best of our knowledge, no sampling-based method has been proposed that is provably safe and persistently feasible.

Model-predictive control (MPC) trajectory planners formulate an optimization program over the vehicle's control inputs over a short time horizon by treating the vehicle's dynamics and environment as constraints.
These approaches typically plan around a reference trajectory, which must be known a priori or generated at runtime; without a reference trajectory, MPC for nonlinear systems and non-convex constraints is typically too slow for real-time application 
\cite{katrakazas2015_motionplanning}.
Depending on the reference trajectory, the MPC problem may be infeasible, but it is unclear how best to generate the reference in arbitrary scenarios.
In structured scenarios, objects such as a road centerline can used as a reference \cite{Frash2013_ACADO_MPC}.
To compensate for uncertainty, Robust MPC has been proposed for linear systems, but vehicle dynamics are nonlinear \cite{Gao2014_robustMPC,orosznonholonomic}.
Methods have also been proposed to ensure \emph{recursive} feasibility for MPC, meaning a solution is available at every planning iteration, for linear systems \cite{lofberg2012mpc} and for nonlinear systems either with two states in discrete time \cite{strief2014robustmpc} or with up to five states with slow dynamics \cite{ma2012robustmpc}; each of these has a tradeoff between discretization granularity and solve time.
For nonlinear systems, MPC approaches typically approximate the dynamics by linearization, or with polynomials \cite{gpopsii}.
To solve the trajectory planning problem, MPC requires discretizing the vehicle's control inputs and trajectory so that each discrete point can be treated as a decision variable \cite{katrakazas2015_motionplanning}.
However, to the best of our knowledge, there is no provable method for discretization to ensure the MPC solution compensates for tracking error and obstacle avoidance.

Reachability-based approaches attempt to ensure safety and persistent feasibility by computing a reachable set to capture a family of possible trajectories given a model of the vehicle's uncertainty.
One such technique, the funnel libraries approach, precomputes reachable ``funnels'' of the vehicle's tracking error around a pre-defined library of trajectories, then links these funnels end-to-end to perform trajectory planning at runtime \cite{majumdar2016funnel}.
Recently, a Hamilton-Jacobi based approach was proposed to compute the reachable set of a vehicle's tracking error to produce lookup table of controllers that bound the tracking error of the vehicle at runtime; this approach requires that the vehicle tracks a reference trajectory produced by, e.g. RRT or MPC \cite{herbert2017fastrack}.
Sums-of-Squares (SOS) programming and control barrier functions have been similarly applied to bound tracking error to ensure safety about a reference trajectory while enabling real-time planning \cite{singh2018sosrealtime,chen2018obstacle}.
For each of these methods, it is unclear how to either ensure that a solution remains persistently feasible or how to represent obstacles so that collision checking operates in real time without sacrificing safety. 
For example, intersecting reachable sets with obstacles represented as semi-algebraic sets can be too slow in practice \cite[Section 6.1]{RTD_big_paper}.
The aforementioned RTD approach achieves safety by computing a Forward Reachable Set (FRS) that includes the vehicle's tracking error and ensures that the vehicle always has a braking trajectory available \cite{kousik2017safe}.
RTD is persistently feasible because it enables the user to enforce a timeout on the online trajectory planning without sacrificing safety; and, real-time planning is possible with a prescription for how to discretize obstacles without losing safety guarantees \cite{RTD_big_paper}.
However, thus far, RTD has only been applied to small mobile robots.

The contribution of this work is demonstrating that RTD is applicable to passenger vehicles with nonlinear dynamics describing the powertrain, chassis, and tires.
We compute an FRS for such a car simulated in CarSim \cite{carsim}, and use RTD to perform safe and persistently feasible maneuvers around static obstacles, as depicted in Figure \ref{fig:carsim_and_matlab_intro}.

\subsection{RTD Overview and Paper Organization}\label{subsec:rtd_overview}

RTD uses a high-fidelity model of the vehicle (Section \ref{sec:high-fidelity_model}) to track desired trajectories in a lower-dimensional subspace (Section \ref{sec:traj-prod_and_track_models}).
An FRS is computed for trajectories of the high-fidelity model in the lower-dimensional subspace by accounting for tracking error (Section \ref{sec:FRS_computation}).
To enable real time operation, obstacles are represented as discretized, finite sets, while still ensuring safety; and to ensure persistent feasibility, specifications are placed on the braking behavior of the vehicle (Section \ref{sec:safety_and_pers_feas}).
Finally, the FRS is used at run-time to map obstacles to constraints for an online optimization step that ensures RTD can only pick safe trajectory plans in a receding horizon fashion (Section \ref{sec:online_planning}).
The method is applied to a CarSim vehicle model (Section \ref{sec:sim_results}).

\subsection{Notation}\label{subsec:notation}

For a set $A$, its boundary is $\partial A$ and its complement is $A^C$, its interior is $\text{int}(A)$, and its power set is $\P(A)$.
The degree of a polynomial is the degree of its largest multinomial; the degree of the multinomial $x^\alpha,\,\alpha\in \N$ is $|\alpha|=\|\alpha\|_1$.
The set $\R_{\geq 0}$ is $[0,\infty)$.
If $z$ is a state, then $\dot{z}$ is its time derivative.
Subscripts denote the index or subspace to which a state belongs.
\section{High-Fidelity Model}\label{sec:high-fidelity_model}

This paper implements RTD on a passenger car model in CarSim.
The inputs are throttle, steering wheel angle, and brake master cylinder pressure.
We say \emph{vehicle} to refer to the Carsim model.
A \emph{high-fidelity model} is used to predict the motion of the vehicle and design a trajectory tracking controller.
Denote the state of the high-fidelity model as $z\hi \in Z\hi \subset \R^{n\hi}$, with dynamics $\dot{z}\hi: [0,T]\times Z\hi \times U \to \R^{n\hi}$.
Initial conditions for these dynamics occupy the space $Z\hio \subseteq Z\hi$; $T$ is the \emph{time horizon} of each trajectory plan; and the control input is drawn from $U \subset \R^{n_U}$.
We use a bicycle model similar to \cite[(1)]{liniger2015optimization} as the high-fidelity model:
\begin{align} \label{eq:high-fidelity_model}
   \dot{z}\hi= \frac{d}{dt}\begin{bmatrix} x_c \\ y_c \\ x \\ y \\ \theta\\ v_x \\ v_y\\ \omega \end{bmatrix}=
   \begin{bmatrix}
   v_x\cos\theta-v_y\sin\theta\\
    v_x\sin\theta+v_y\cos\theta\\
    v_x\cos\theta-v_y\sin\theta - \omega(y - y_c) \\
    v_x\sin\theta+v_y\cos\theta + \omega(x - x_c) \\
    \omega\\
    \frac{1}{m}F_x-\frac{1}{m}F_{\regtext{f},y}\sin \delta+v_y\omega\\
     \frac{1}{m}F_{\regtext{f},y}\cos \delta+\frac{1}{m}F_{\regtext{r},y}-v_x\omega\\
     \frac{l_f}{I_z}F_{\regtext{f},y}\cos \delta -\frac{l_r}{I_z}F_{\regtext{r},y}
   \end{bmatrix},
\end{align}
where $x_c$ and $y_c$ are the position of the vehicle's center of mass; $x$ and $y$ are the position of any point on the vehicle's body; $\theta$ is the vehicle's heading in the global coordinate frame; $v_x$, $v_y$ are longitudinal and lateral speed of the center of mass; and $\omega$ is yaw rate.
The constants $m$, $I_z$, $l_f$, and $l_r$ are the vehicle's mass, yaw moment of inertia, distance from the front wheel to center of mass, and distance of the rear wheel to center of mass.
To identify the model parameters, the vehicle is run through a series of open-loop acceleration, and deceleration inputs.
We fit polynomials relating the throttle and brake inputs to the driving force, $F_x$, and we find a linear relationship between wheel angle, $\delta$, and steering wheel angle.
Cornering maneuvers produce data to fit a simplified Pajecka tire model \cite[(2a, 2b)]{liniger2015optimization} to the lateral tire forces, $F_{\regtext{f},y}$ and $F_{\regtext{r},y}$.
Since $F_x$, $F_{\regtext{f},y}$, and $F_{\regtext{r},y}$ are continuous, the dynamics \eqref{eq:high-fidelity_model} are continuous.

Recall that \eqref{eq:high-fidelity_model} cannot perfectly capture the motion of the vehicle.
However, since the time horizon $[0,T]$, is compact, we can bound prediction error as follows.
\begin{assum}\label{ass:predict_within_epsilon}
Future state predictions given by the high-fidelity model \eqref{eq:high-fidelity_model} predict each state of the vehicle within an error bound $\vep_i > 0$ for $i=1,\cdots,n\hi$ at each time $t \in [0,T]$.
\end{assum}
\noindent  By this assumption, the high-fidelity model lies within $\vep_x, \vep_y$ of the vehicle in its $x$ and $y$ coordinates, as required by \cite[Assumption 9]{RTD_big_paper}.
We simulate the high-fidelity model and compare its state to Carsim data to empirically find the error bounds: $|\vep|\leq [0.1,\, 0.1,\, 0.12,\, 0.15,\, 0.02,\, 0.4,\, 0.08,\, 0.05]^\top$ where $|\cdot|$ is taken elementwise.

Notice that the dynamics of all points on the vehicle's body are included in \eqref{eq:high-fidelity_model}.
This is because the vehicle has nonzero volume, so it is insufficient to only consider the center-of-mass dynamics for trajectory planning.
Approaches that only plan with the center of mass typically expand the size of obstacles in the vehicle's environment such that, if the center of mass does not lie within expanded obstacles, then no point on the vehicle's body can lie in the actual obstacle \cite{elbanhawi2014sampling}.
In such approaches either the vehicle's footprint is a disk, so that obstacles can be expanded uniformly; or, the obstacle representation requires tuning parameters in a trial-and-error fashion to compensate for the vehicle's shape and possible range of headings, making it hard to ensure safety.
In contrast, RTD directly addresses planning trajectories with the vehicle's entire body.

\begin{rem}\label{rem:X_and_X0_rigid_body}
The state space $Z\hi$ has a two-dimensional \emph{spatial subspace} $X \subset Z\hi$ with coordinates $x$ and $y$.
The vehicle has a rectangular \emph{footprint} $X_0 \subset X$ that represents all points on the vehicle's body at the beginning of each planning iteration.
The states $x_c$ and $y_c$ evolve in a \emph{center-of-mass subspace} $X_c \subset X$.
According to the dynamics of $x$ and $y$ in \eqref{eq:high-fidelity_model}, the vehicle's footprint acts as a rigid body \cite[Lecture 7]{dynamics_MIT_OCW}.
\end{rem}
 \section{Producing and Tracking Trajectories}\label{sec:traj-prod_and_track_models}

Since the high-fidelity vehicle model is nonlinear with saturating inputs, it is difficult to use for planning in real-time.
Instead, RTD plans \emph{desired trajectories} with a lower-dimensional \emph{trajectory-producing model}, which has \emph{shared states} $z \in Z \subset Z\hi$, where $\dim(Z) = n_Z < \dim(Z\hi)$.
The model includes \emph{trajectory parameters}, $k$, that are drawn from a \emph{parameter space}, $K$.
The trajectory-producing model produces desired trajectories with dynamics  $\dot{z}\des: [0,T]\times Z\times K \to \R^{n_Z}$ with a space $Z_0 \subset Z\hio$ of initial conditions.
We use the following trajectory-producing model:
\begin{align}
    \dot{z}\des(t,z(t),k) = \begin{bmatrix} \dot{x}(t) \\ \dot{y}(t) \end{bmatrix} &= \begin{bmatrix} k_2-k_1(y(t)-y_c(0))\\ v_y^*+k_1(x(t)-x_c(0)) \end{bmatrix}\label{eq:traj-producing_model}\\
    v_y^*&=k_1\left(l_r-\frac{m\,l_f}{C_r\,(l_r+l_f)}k_2^2\right)\label{eq:steady_state_vy},
\end{align}
where $z = [x,y]^\top$; $k_1$ (resp. $k_2$) specifies a constant desired yaw rate (resp. longitudinal speed); and $C_r$ is the rear cornering stiffness from the tire force model in \eqref{eq:high-fidelity_model}.
The lateral speed, $v_y^*$, is derived from steady-state, linear tire force assumptions \cite[Section 10.1.2]{schramm2014vehicle}.
Notice that \eqref{eq:traj-producing_model} only has the two states $x$ and $y$, i.e. the center of mass and heading dynamics are omitted.
This is because the desired trajectories of the states $x_c,\ y_c$, and $\theta$ are treated as functions of the parameters $k$ and time, which lets us compute explicit solutions for their trajectories.
Consequently, \eqref{eq:traj-producing_model} produces trajectories of the vehicle's entire footprint in $X$, with initial conditions anywhere in the footprint $X_0$.
So, the shared state subspace $Z \subset Z\hi$ is in fact the spatial subspace $X$.

We use the trajectory-producing model as follows.
For every trajectory parameter $k \in K$, the high-fidelity model generates a feedback controller $u_k: [0,T] \times Z\hi \to U$ that attempts to track the trajectory parameterized by $k$; to shorten vocabulary, when applying $u_k$, we say that the vehicle \emph{tracks $k$}.
In our case, the vehicle uses linear MPC to track trajectories, implemented with MATLAB's MPC toolbox.
MPC was chosen to incorporate input saturation and rate limits, but any feedback controller can be used with RTD.

We now address the fact that desired trajectories produced by \eqref{eq:traj-producing_model} are not necessarily dynamically feasible for the high-fidelity model \eqref{eq:high-fidelity_model}.
This is because the dynamics and dimension of the two models differ, and because state estimation error can accumulate as a trajectory is tracked using feedback.
We refer to the difference between the high-fidelity model and the trajectory-producing model as \emph{tracking error}.
We bound tracking error with a function $g: [0,T]\times Z\times K \to \R^{n_Z}$, and use it to create a \emph{trajectory-tracking model} that matches the desired trajectories to the high-fidelity model; this requires the following assumption.
\begin{assum}\label{ass:compact_sets_and_cont_dyn}
The spaces $Z\hi$, $Z\hio$, $U$, $Z$, $Z_0$, and $K$ are compact subsets of Euclidean space that admit semi-algebraic representations.
\end{assum}

\begin{figure}
    \centering
    \includegraphics[width=0.95\columnwidth]{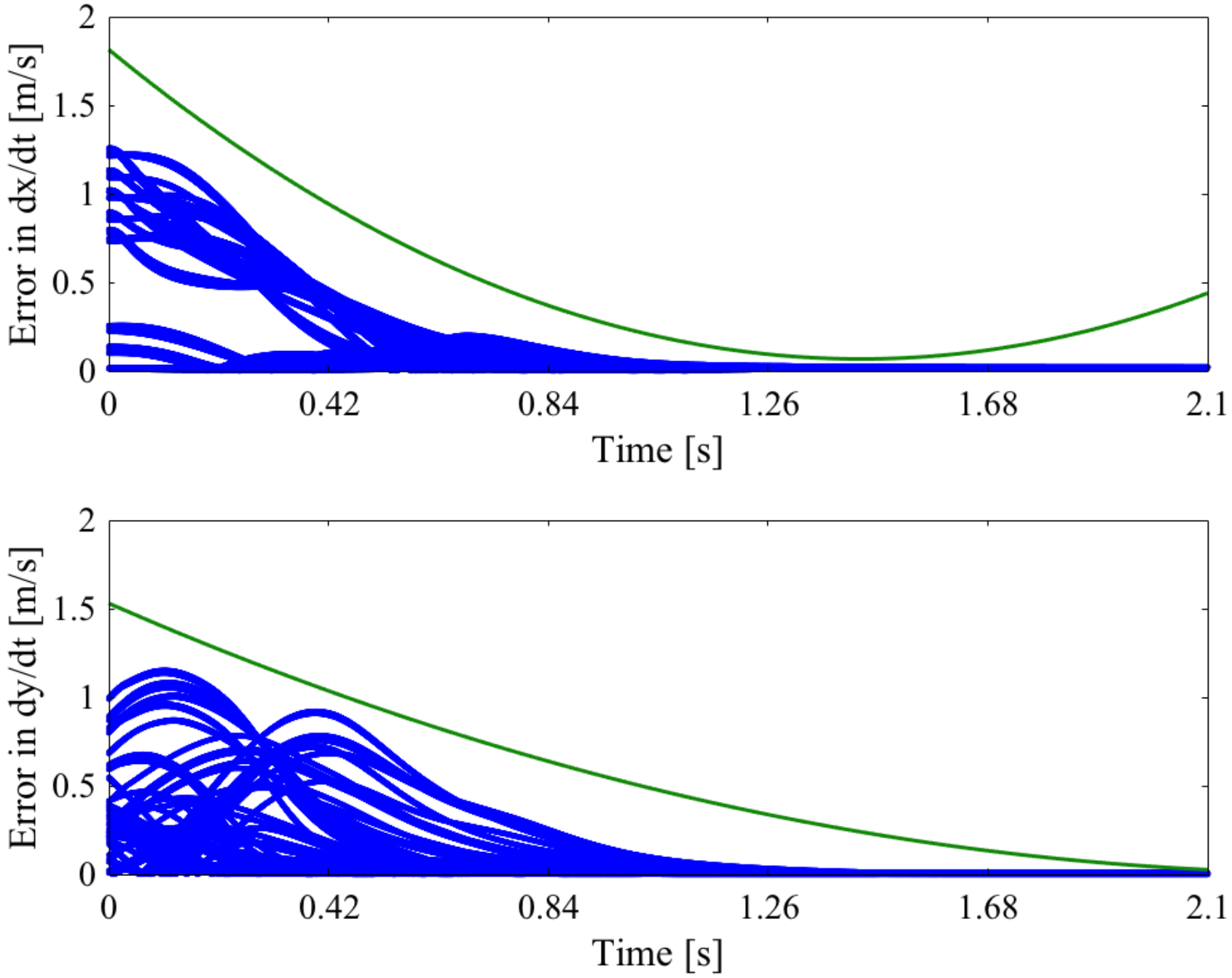}
    \caption{Example of tracking error plotted for a reference trajectory of $k_1=0$ rad/s and $k_2=12$ m/s and $T=2.1$ s.
    The top and bottom plots show the time derivative of absolute error in $x$ and $y$ states respectively.
    Data (blue) is from CarSim and captures initial velocities and yaw rates between  10.78 to 13.26 ms and -0.25 to 0.25 rad/s.
    The green lines are the error functions $g_x$ and $g_y$ as in \eqref{eq:g_tracking_error_defn}.}
    \label{fig:actual_data_g}
\end{figure}

Now, for the states $x$ and $y$ in \eqref{eq:traj-producing_model}, we find functions $g_x, g_y: [0,T]\times K \to \R_{\geq 0}$ such that
\begin{align}
   \max_{z\hi \in A\hi} |z\hii(t,z\hi,u_k) - z_{\regtext{des},i}(t,z,k)| \leq \textstyle\int_0^t g_i(\tau,k)d\tau,\label{eq:g_tracking_error_defn}
\end{align}
for all $t \in [0,T]$, $z \in Z$, and $k \in K$.
The subscript $i = x, y$ selects the corresponding components; the set $A\hi = \{z\hi \in Z\hi~\mid~z\hii = z_i~\regtext{for}~i = x, y\}$.
Arguments to $z\hi$, $u_k$, and $z$ are dropped to lighten notation.
The \emph{tracking error function} is $g = [g_x,g_y]^\top$.
In this work, $g_x$ and $g_y$ are polynomials of degree 2 that overapproximate tracking error data found by simulating the high-fidelity model tracking reference trajectories from a variety of initial conditions.
Importantly, state estimation error in $z\hio$ is added to the initial conditions, so $g$ conservatively approximates the prediction errors described in Assumption \ref{ass:predict_within_epsilon}.
Although not considered in this paper, state estimation error due to imperfect sensors and observer design can also be accounted for in $g$.
Figure \ref{fig:actual_data_g} shows data collected in CarSim of the vehicle tracking a reference trajectory of 12 m/s and 0 rad/s from initial velocities between 10.78 to 13.26 m/s and yaw rates between -0.25 and 0.25 rad/s, along with the computed error functions.
Constructing $g$ is not the focus of this work, but it can be conservatively approximated with, e.g., SOS programming \cite[Chapter 7]{lasserre2009moments}.

The tracking error function lets the trajectory-producing model ``match'' the high-fidelity model in the shared states:
\begin{lem}\label{lem:hi-fid_matches_traj-prod}
Let $L_d = L^1([0,T],[-1,1]^{n_Z})$ denote the space of absolutely integrable functions from $[0,T]$ to $[-1,1]^{n_Z}$ and recall that $n_Z = 2$ in \eqref{eq:traj-producing_model}.
Let $z\hi: [0,T] \to Z\hi$ denote a trajectory of \eqref{eq:high-fidelity_model} from arbitrary $z\hio \in Z\hio$ and tracking arbitrary $k \in K$.
Then, there exists $d \in L_d$ such that, almost everywhere $t \in [0,T]$,
\begin{align}\label{eq:error_equation}
    z\hii(t) = z_{\regtext{hi},0,i} + \textstyle\int_0^t \left(z_{\regtext{des},i}(\tau,z,k) + g_i(\tau,k)\cdot d_i(\tau)\, \right)d\tau
\end{align}
where $i = x, y$ selects each shared state in $Z \subset Z\hi$.
\end{lem}
\begin{proof}
From \eqref{eq:g_tracking_error_defn}, $g$ bounds the maximum absolute error in $x$ and $y$ that can accumulate by any $t\in[0,T]$. 
Taking both the positive and negative case of the term inside the absolute value on the left hand side, then taking time derivatives, gives us the following inequalities: $\dot{z}\hii(t,z\hi,u_k) \leq \dot{z}_{\regtext{des},i}(t,z,k)+g_i(t,k)$ and
$\dot{z}\hii(t,z\hi,u_k) \geq \dot{z}_{\regtext{des},i}(t,z,k)-g_i(t,k)$.
Therefore, for all $t\in[0,T]$, we can pick $d_i(t)\in[-1,1]$ and integrate from $0$ to $t$ such that \eqref{eq:error_equation} is satisfied.
\end{proof}

We now define the \emph{trajectory-tracking model}, with state $z = [x,y]^\top \in Z$, as:
\begin{align}
    \dot{z}_i(t) = \dot{z}_{\regtext{des},i}(t,z(t),k) + g_i(t,k)\cdot d_i(t),\label{eq:traj-tracking_model}
\end{align}
where $i = x, y$ and $d = [d_x,d_y]^\top \in L_d$.
The utility of \eqref{eq:traj-tracking_model} is that, by Lemma \ref{lem:hi-fid_matches_traj-prod}, it can match any trajectory of the high-fidelity model in the shared states over the time horizon $[0,T]$; and, since the shared states occupy a lower-dimensional space than $Z\hi$, we can compute an FRS of the trajectory-tracking model.
\section{The Forward-Reachable Set}\label{sec:FRS_computation}
We now briefly discuss the FRS, detailed in \cite[Section 3]{RTD_big_paper}. The FRS is all points in $X$, and associated parameters in $K$, that are reachable by the trajectory-tracking model:
\begin{align}\begin{split}\label{eq:Xfrs}
    \X\frs = \{&(\hat{z},\hat{k}) \in X\times K~\mid~\exists~t \in [0,T],\ z_0 \in Z_0,\ d \in L_d\\
    &\regtext{s.t.}\ \dot{z}_i(\tau) = \dot{z}_{\regtext{des},i}(\tau,z(\tau),\hat{k}) + g_i(\tau,\hat{k})\cdot d_i(\tau),\\
    &z(0) = z_0,\ \regtext{and}\ z(t) = \hat{z} \},
\end{split}\end{align}
where $i = x, y$.
By Lemma \ref{lem:hi-fid_matches_traj-prod}, $\X\frs$ contains all points in $X$ that are reachable by the high-fidelity model tracking any trajectory parameterized by any $k \in K$.

\begin{rem}\label{rem:w_geq_1_in_Xfrs}
Recall that the dynamics $\dot{z}\des$ from \eqref{eq:traj-producing_model} and $g$ from \eqref{eq:g_tracking_error_defn} are polynomials.
Furthermore, the spaces $[0,T],\ Z$, and $K$ are compact and admit semi-algebraic representations by Assumption \ref{ass:compact_sets_and_cont_dyn} 
Therefore, by \cite[Lemma 14 and Remark 18]{RTD_big_paper}, we can use SOS programming to find a polynomial $w_\al: X\times K \to \R$ of degree $2\al \in \N$ for which
\begin{align}
   \X\frs \quad\subseteq\quad \{(z,k) \in X\times K~|~ w_\al(z,k) \geq 1\},
\end{align}
that is, the 1-superlevel set of $w_\al$ overapproximates the FRS.
\end{rem}
\noindent We use Remark \ref{rem:w_geq_1_in_Xfrs} to ensure safety and persistent feasibility of the vehicle in Section \ref{sec:safety_and_pers_feas}.
See \cite[Theorem 6]{majumdar2014convex} for a proof, and \cite[Program $(D^l)$]{RTD_big_paper} to compute $w_\al$. 

In Section \ref{sec:sim_results}, the vehicle will be run on a test track with a max speed of 15 m/s.
We compute six FRSes (i.e. six $w_\al$ polynomials) for commanded velocities of 3--5, 5--7, 7--9, 9--11, 11--13, and 13--15 m/s.
The vehicle only goes below 3 m/s when braking to a stop.
Each $w_\al$ has $\al = 6$.
We use polynomial error functions $g$ of degree 2; we fit a $g$ for each FRS.
Computing more FRSes with finer ranges of initial conditions and commands, would reduce conservatism; however, we found empirically that using 6 FRSes led to the vehicle completing the experiments of Section \ref{sec:sim_results}.
Any number of FRSes could be used as long as, together, they cover the vehicle's initial conditions.
Per \cite[Appendix 14.2]{RTD_big_paper}, at each planning iteration, the vehicle picks the FRS with the highest possible commanded speed that contains its current initial condition, which maintains the safety guarantee presented in the following section.
\section{Safety and Persistent Feasibility}\label{sec:safety_and_pers_feas}

We now define safety and persistent feasibility by using $w_\al$ from Remark \ref{rem:w_geq_1_in_Xfrs} to project the FRS into $X$ and $K$.
Note that we have computed multiple FRS's for separate speed ranges in Section \ref{sec:FRS_computation}; the material in this section holds for each FRS independently.

\subsection{Ensuring Safety}\label{subsec:obs_rep_for_safety}
We ensure safety by representing obstacles with discrete points in $X$ that become nonlinear constraints for the online trajectory planner.
This representation lets the trajectory planner run in real time \cite[Section 6.1]{RTD_big_paper}.
First, we specify how obstacle data must be received from sensors:

\begin{assum}\label{ass:obs_and_D_sense}
Obstacles, denoted $X\obs \subset X$, are closed polygons that are static with respect to time.
There are at most $n\obs \in \N$ obstacles within the vehicle's \emph{sensor horizon} distance $D\sense > 0$ at any time.
The vehicle senses all obstacles within $D\sense$ of its center of mass.
Obstacles do not appear spontaneously within the sensor horizon.
\end{assum}
\noindent Note that $X\obs$ can contain more than one disjoint polygonal obstacle.
This assumption is reasonable for obstacles represented by occupation grids or line segments fit to lidar data.
Occlusions can be treated as static obstacles. 

To relate obstacles to trajectory parameters, we define the \emph{FRS parameter projection map} $\pi_K: \P(X) \to \P(K)$ for which
\begin{align}
    \pi_K(X') = \{k \in K~\mid~\exists~z \in X'~\regtext{s.t.}~w_\al(z,k) \geq 1\}.\label{eq:pi_K}
\end{align}
Then, the safe set of parameters corresponding to an obstacle $X\obs \subset X$ is $K\safe = \pi_K(X\obs)^C$.
To conservatively approximate $K\safe$ in real time, we use the approach proposed in \cite[Section 6]{RTD_big_paper}, wherein $X\obs$ is \emph{buffered}, then its boundary is \emph{discretized}.
For a chosen buffer distance $b \geq 0$, the \emph{buffered obstacle} is:
\begin{align}
    X\obs^b = \left\{z \in X~\mid~\exists~z' \in X\obs~\regtext{s.t.}~\norm{z - z'}_2 \leq b\right\}. \label{eq:buffered_obs}
\end{align}
Since $X\obs$ is a polygon, the boundary of $X\obs^b$ consists of a finite number of line segments and circular arcs of radius $b$ \cite[Section 9.2]{minkowski_sum_fogel}.
Let $L = \{L_1,\cdots,L_{n_L}\}$ and $A = \{A_1,\cdots,A_{n_A}\}$ denote the sets of line segments and arcs, respectively, so that $\bd X\obs^b = \left(\bigcup_{i = 1}^{n_L}L_i\right) \cup \left(\bigcup_{i = 1}^{n_A}A_i\right)$.
It is shown in \cite[Section 6]{RTD_big_paper} that these line segments and arcs $\bd X\obs^b$ can be sampled to produce a finite \emph{discretized obstacle} $X_p \subset X$ that conservatively approximates the obstacle, i.e. $\pi_K(X_p)^C \subseteq K\safe$.
To understand why $X_p \subset \bd X\obs^b$, note that the spatial component of the vehicle's dynamics is continuous, so the vehicle cannot collide with the obstacle without passing through the obstacle's boundary first.
Example buffered and discretized obstacles are shown in Figure \ref{fig:braking}.

We construct $X_p$ as follows, summarizing \cite[Algorithm 1]{RTD_big_paper}.
Given a connected, compact curve $S: [0,1] \to \R^2$ and a distance $s > 0$, let \texttt{sample}$(S,s)$ return a (finite) set $P$ of points spaced along $S$ such that, for any point $p \in P$, there exists at least one other point $p' \in P$ no farther than $s$ away in the 2-norm, i.e. $\norm{p - p'}_2 \leq s$.
We also require that $S(0), S(1) \in P$, i.e. the ``endpoints'' of $S$ are in $P$.
Then,
\begin{align}
    X_p = \left(\bigcup_{i=1}^{n_L}\texttt{sample}(L_i,s_L)\right)\cup \left(\bigcup_{i=1}^{n_A}\texttt{sample}(A_i,s_A)\right),\label{eq:construct_X_p}
\end{align}
where $s_L > 0$ is the \emph{point spacing} and $s_A > 0$ is the \emph{arc point spacing}.
By construction, $X_p$ is a set of points that ``surround'' the obstacle.
So, to ensure safety, we must guarantee that no point on the vehicle's footprint can travel ``between'' any pair of adjacent points in $X_p$ farther than a chosen buffer distance $b$, otherwise the vehicle can collide with $X\obs$.
This means that $s_L$ and $s_A$ must be small enough that points in $X_p$ are close to each other in the 2-norm.
The values of $s_L$ and $s_A$ depend upon the shape of the vehicle's footprint $X_0$; which, in this work, is a rectangle (as per Remark \ref{rem:X_and_X0_rigid_body}).
We ensure safety with the following lemma that follows from \cite[Example 66 and Theorem 68]{RTD_big_paper}.

\begin{lem} \label{lem:X_p_is_safe}
Suppose $X\obs$ is a set of obstacles as in Assumption \ref{ass:obs_and_D_sense}, and $X_0$ has width $W$.
Pick a buffer distance $b \in (0,W/2)$.
Set $s_L = 2b$ and $s_A = 2b\sin(\pi/4)$.
Then, if $X_p$ is constructed as in \eqref{eq:construct_X_p}, the unsafe parameters corresponding to $X_p$ are a conservative approximation of the unsafe parameters corresponding to $X\obs$, i.e. $\pi_K(X_p) \supseteq \pi_K(X\obs) = K\safe^C$.
\end{lem}
\noindent By Lemma \ref{lem:X_p_is_safe}, the vehicle is safe over a time horizon $[0,T]$, i.e. the duration of a single planned trajectory, when tracking any $k \in \pi_K(X_p)^C$.
For the vehicle, we pick $b = 0.05$ m, so $s_L = 0.1$ m and $s_A = 0.07$ m.

\subsection{Ensuring Persistent Feasibility}\label{subsec:pers_feas}

\begin{figure}
    \centering
    \includegraphics[width=1.0\columnwidth]{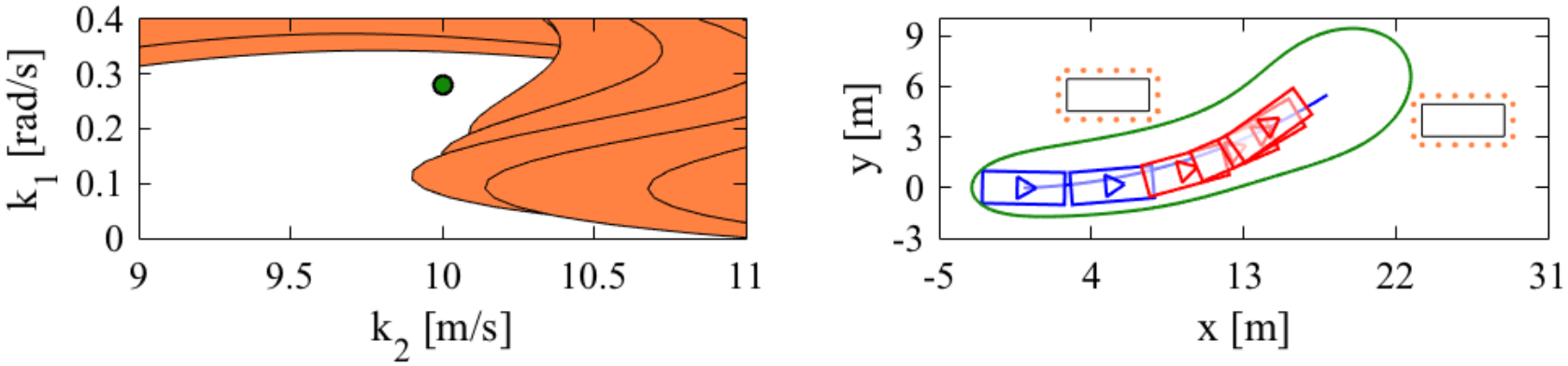}
    \caption{Examples of safe braking along a trajectory, of discretized obstacles, and of the maps $\pi_K$ and $\pi_X$. 
    The left subplot is $K$.
    The green circle indicates the parameters for the selected trajectory, $k = (0.28~\mathrm{rad/s},\ 10~\mathrm{m/s})$.
    The filled orange contours are the parameter projection map, $\pi_K(X_p)$, for the orange obstacle points in the right subplot.
    The right subplot is $X$.
    The initial speed and yaw rate of the vehicle are $0.027$ rad/s and $11.1$ m/s.
    The vehicle's pose, taken from CarSim data, is plotted every 0.5 s; it is blue when initially tracking the trajectory, and red when braking.
    The green contour is the spatial projection map of the forward reachable set, $\pi_X(k)$.
    Notice that the FRS does not intersect the obstacles, meaning that the chosen $k$ is safe for the actual vehicle to track.}
    \label{fig:braking}
\end{figure}

We now ensure the vehicle is able to always find a safe trajectory while planning with a receding-horizon strategy by specifying a minimum duration $T$ for each planned trajectory and a minimum sensor horizon $D\sense$.
\begin{rem}
We assume that sensor data is processed and passed to the trajectory planner instantaneously.
Then, trajectory planning is limited to a duration $\tau\plan \in (0,T)$ every planning iteration; a new $k \in K$ is found every $\tau\plan$ seconds, otherwise the vehicle begins braking.
Note that, in practical applications, $\tau\plan$ can be increased to include the time it takes to process sensor data.
\end{rem}

\noindent In this work, $\tau\plan = 0.5$ seconds.
We now consider how the vehicle must brake to be safe.
To understand how the vehicle tracks parameters, we define the \emph{FRS spatial projection map} $\pi_X: \P(K) \to \P(X)$ for which
\begin{align}
    \pi_X(K') = \{z \in X~\mid~\exists~k \in K'~\regtext{s.t.}~w_\al(z,k) \geq 1\}\label{eq:pi_X}.
\end{align}
By Remark \ref{rem:w_geq_1_in_Xfrs}, for any $k \in K$, $\pi_X(k) \subset X$ contains all points in $X$ reachable by any point on the vehicle's body, with dynamics \eqref{eq:high-fidelity_model}, over the time horizon $[0,T]$.
Now, we use $\pi_X$ to define safe braking behavior.
\begin{assum}\label{ass:brake}
While tracking any $k \in K$ over the time horizon $t \in [0,T]$, if the vehicle begins braking at $t = \tau\plan$, then all points on the vehicle's body lie within $\pi_X(k)$, the spatial projection of the parameter $k$.
\end{assum}
\noindent In other words, the FRS is large enough that the vehicle can brake within it for any trajectory parameter; as discussed next, this requires choosing $T$ so the vehicle can satisfy Assumption \ref{ass:brake} as shown in Figure \ref{fig:braking}.

To ensure Assumption \ref{ass:brake} can be fulfilled, the time horizon $T$ must be large enough that, when tracking any $k \in K$, the vehicle travels farther than the maximum braking distance it can achieve by braking at $\tau\plan$ from any state $z\hi \in Z\hi$ that results from tracking $k$.
Notice that \eqref{eq:traj-producing_model} creates trajectories that maintain a fixed speed over $[0,T]$.
Furthermore, the vehicle has a maximum stopping distance of $D\stp$ at the max speed considered in the FRS; and, the vehicle's stopping distance increases with the square of its speed.
Therefore, by \cite[Remark 73]{RTD_big_paper}, picking
\begin{align}
    T \geq \tau\plan + D\stp / \vmax
\end{align}
ensures that the trajectories in the FRS are long enough to satisfy Assumption \ref{ass:brake} \cite[Appendix 14]{RTD_big_paper}.
For example, for an FRS with $\vmax=11$ and $\tau\plan=0.5$ s, $D\stp=15.4$ m and $T=1.9$ s.
Note that this assumes no delay in braking actuation; if there is delay, then $T$ must be increased to include it.
Also, recall that, since we compute multiple FRSes, the max speed of each FRS determines its time horizon $T$.

Finally, to ensure persistent feasibility, recall by Remark \ref{ass:predict_within_epsilon} that the vehicle's state estimation error is bounded in $x$ and $y$ by $\vep_x$ and $\vep_y$ respectively.
To compensate for this error, we expand obstacles by $\vep_x$ in $x$ and $\vep_y$ in $y$ before creating the discretized obstacle with \eqref{eq:construct_X_p}.
Recall that state estimation error, and the resulting tracking error, is accounted for in $g$ as in \eqref{eq:g_tracking_error_defn}.
We conclude this section by specifying the sensor horizon required for persistent feasibility.

\begin{lem}\label{lem:T_sense_pers_feas}
\cite[Theorem 35]{RTD_big_paper}.
Let $X\obs \subset X$ be obstacles as in Assumption \ref{ass:obs_and_D_sense}.
Let $\vep = \sqrt{\vep_x^2 + \vep_y^2}$.
Let $\vmax$ denote the vehicle's max speed.
Suppose that the vehicle has known safe $k_0 \in K$ at $t = 0$.
Then, if the sensor horizon is
\begin{align}
    D\sense \geq (T+\tau\plan)\cdot\vmax + 2\vep,    
\end{align}
the vehicle can always either find a new trajectory parameter $k \in K$ or begin braking safely at every $t = j\tau\plan$ where $j \in \N$.
\end{lem}
\noindent See \cite[Theorem 35]{RTD_big_paper} for the proof.
In this work, $D\sense \geq 42.4$ m; note that this is within the reported range of many commercial lidar units such as \cite{velodyne2007whitepaper}.
\section{Online Planning}\label{sec:online_planning}

We now apply RTD's online trajectory optimization (see \cite[Section 7 Algorithm 2]{RTD_big_paper}) to the vehicle.
Per Lemma \ref{lem:T_sense_pers_feas}, at each planning iteration, the vehicle either finds a new safe plan (i.e., picks a new $k \in K$) or begins braking safely.
First, we state how to find a new safe plan.

Let $w_{\al}$ be as in Remark \ref{rem:w_geq_1_in_Xfrs}.
Suppose the vehicle is at planning iteration $j \in \N$ and tracking the previous iteration's safe parameter $k_{j-1}$.
Suppose that $X\obs \subset X$ is an obstacle as in Assumption \ref{ass:obs_and_D_sense}, sensed as in Lemma \ref{lem:T_sense_pers_feas}.
Let $J: K \to \R$ be an arbitrary cost function.
Let $X_p$ be the discretized obstacle constructed as in Lemma \ref{lem:X_p_is_safe}.
We find $k_j$ with the program:
\begin{align}\label{prog:OptK}
    k_j = \underset{k}{\regtext{argmin}} \left\{ J(k)~\mid~w_{\al}(z,k) < 1 ~\forall~z\in X_p. \right\}
\end{align}
By Remark \ref{rem:w_geq_1_in_Xfrs}, the constraint $w_\al(z,k) < 1$ ensures that, for any feasible $k$, no point on the vehicle's body can reach any $z \in X_p$ at any $t \in [j\tau\plan,(j+1)\tau\plan]$.
By Lemma \ref{lem:X_p_is_safe}, this means that no point on the vehicle's body can reach $X\obs$.

In practice, feasible solutions to \eqref{prog:OptK} can be found quickly because $X_p$ becomes a finite list of point constraints on the decision variable $k$.
We implement \eqref{prog:OptK} with MATLAB's \texttt{fmincon} general nonlinear solver.
The cost function at each planning iteration is the vehicle's position and velocity relative to a desired waypoint and velocity, respectively, which are given by a high-level planner described in Section \ref{subsec:implementation}.
Recall that, by Assumption \ref{ass:brake}, if a new safe $k$ cannot be found within $\tau\plan$, then the car can always brake safely within the FRS of the previous plan.

\section{Simulation}\label{sec:sim_results}
This section compares RTD to Rapidly-Exploring Random Tree (RRT) and Nonlinear Model-Predictive Control (NMPC) trajectory planners.
The simulations are run on a 2.6 GHz computer with 128 GB RAM.
Planning times are reported using Matlab's \texttt{tic} and \texttt{toc} functions.
All planners use a receding horizon strategy where they must plan a new trajectory every $\tau\plan = 0.5$ s.
In the first experiment, all three planners are run with a real-time planning limit enforced.
In the second experiment, RRT and NMPC are given extra time.

The vehicle runs on a $1036$ m, counter-clockwise, closed loop test track with 7 turns (with approximate curvatures of 0.005--0.04 m$\inv$) and two $4$ m wide lanes.
Twenty stationary obstacles (with random length of 3.3--5.1 m length and width of 1.7--2.5 m) are distributed around the track in random lanes, with random longitudinal spacing of 40-55 m along the road.
Each obstacle is placed in the lane center with its heading in the direction of the lane.
We generated ten such random tracks; even though the mean obstacle spacing is the same, the tracks vary in difficulty.
For example, some tracks require the vehicle to perform overtaking maneuvers while cornering.
The vehicle begins each simulation at the northwest corner of the track in the left lane, with first obstacle at least 50 m away.
The obstacle spacing means the vehicle should be able to navigate every test track.
A road-block scenario, where RTD is forced to brake to a stop, is also shown in the video linked in \ref{subsec:results}.

A high-level planner places waypoints ahead of the vehicle at a \emph{lookahead distance} proportional to the vehicle's current speed.
If the lane centerline from the vehicle's current position and lane to the waypoint intersects an obstacle, the waypoint is switched to the other lane to encourage a lane change.
Lane keeping is not explicitly enforced but is encouraged via the cost function.
Each simulation is deemed successful if the vehicle completes one lap of the track.

\subsection{Trajectory Planner Implementations}\label{subsec:implementation}
RTD is implemented as discussed in Sections \ref{sec:high-fidelity_model}--\ref{sec:online_planning}.
Constraints in the online optimization program \eqref{prog:OptK} limit the commanded change in velocity to 1 m/s and yaw rate to less than 0.25 rad/s in each planning iteration.
These constraints create initial condition ranges over which the tracking error functions, described in Section \ref{sec:FRS_computation}, are valid.
FRSes can be computed for more aggressive maneuvers, but, even without them, the vehicle is able to successfully navigate the test track.
To keep the vehicle on the road, RTD buffers the road boundaries by 2.5 m (outside the road) and incorporates these buffers as obstacles.
Since the FRS includes the full vehicle body, this ensures that the vehicle's center of mass stays in the road boundaries when tracking any trajectory planned by RTD.

The RRT planner is implemented based on \cite{kuwata2009rrt}.
To compensate for the vehicle's footprint, obstacles are buffered by 4 m in length and 1.5 m in width.
New nodes are creating by first selecting a random existing node, then forward-integrating the vehicle's high-fidelity model with randomly-chosen control inputs held for 0.5 s.
This creates 50 points spaced 0.01 s apart; the last such point is the new node, which is discarded if any of the points leave the track or enter a buffered obstacle.
Two trees are built in parallel: one with throttle inputs, and one with braking inputs.
The cost at each node is the distance to the current waypoint, plus penalties for being near obstacles or road boundaries, and for commanding large control inputs.

The NMPC planner uses GPOPS-II, a commercially available pseudo-spectral nonlinear MPC solver \cite{gpopsii}.
GPOPS-II uses a kinematic bicycle model, similar to \cite[Section 9.2.3]{RTD_big_paper}, with acceleration and steering wheel angle rate as inputs.
Reducing the number of inputs and complexity of the dynamics was found to reduce solve time.
Obstacles are buffered by 4 m in length and 1.25 m in width. 
The track is discretized and represented as a set of adjacent rectangles.
Constraints are created as half-planes to ensure that the planned trajectory (of the center of mass) does not enter buffered obstacles or exit rectangles defining the track.
To reduce the number of constraints, only obstacles and road boundaries within the lookahead distance of the high-level planner are considered.

\subsection{Results}\label{subsec:results}
\begin{table}[t]
\begin{tabular}{|c|r|r|r|r|c|c|}
\hline
\multirow{2}{*}{Planner} & \multicolumn{2}{l|}{Planning Time  (s)} & \multicolumn{2}{l|}{\begin{tabular}[c]{@{}l@{}}\% of Track \\ Complete\end{tabular}} & \multirow{2}{*}{Crashes} & \multirow{2}{*}{\begin{tabular}[c]{@{}l@{}}Safe \\ Stops\end{tabular}} \\ \cline{2-5}
 & Avg & \multicolumn{1}{c|}{Max} & Avg & Max &  &  \\ \hline
RTD & \cellcolor{Gray}\textbf{0.09} &  \cellcolor{Gray}0.50 &  \cellcolor{Gray}\textbf{100} &  \cellcolor{Gray}100 &  \cellcolor{Gray}0 &  \cellcolor{Gray}0 \\ \hline
\multirow{2}{*}{RRT} & 5.00 & 5.00 & 31 & 86 & 0 & 10   \\ \cline{2-7} &  \cellcolor{Gray}0.50 &  \cellcolor{Gray}0.50 &  \cellcolor{Gray}13 &  \cellcolor{Gray}38 &  \cellcolor{Gray}1 &  \cellcolor{Gray}9
\\ \hline
\multirow{2}{*}{GPOPS}  & 3.71 & 72.58 & \textbf{100} & 100 & 0 & 0\\ \cline{2-7} &  \cellcolor{Gray}0.50 & \cellcolor{Gray} 0.50 & \cellcolor{Gray}0 &  \cellcolor{Gray}0 & \cellcolor{Gray}0 &  \cellcolor{Gray}10 
 \\ \hline
\end{tabular}
\caption{Simulation results comparing RTD, RRT, and GPOPS-II on 10 simulated tracks.
The first experiment, with the real-time planning limit, is shown in gray, and the second experiment in white.
The fourth and fifth columns show the average and max percent of each track completed.
The six and seventh columns count the number of crashes or safe stops if the vehicle did not complete the track.}
\label{tab:sim_results}
\vspace{-3mm}
\end{table}

Results are shown in Table \ref{tab:sim_results}.
A video of RTD planning in real time is available at \url{http://www.roahmlab.com/acc2019_rtd_video}.
In the first experiment, RTD successfully navigates the track in all 10 trials, with an average planning time of 0.086 s.
With the 0.5 s time limit, RRT on average, navigates 13\% of the track on average.
When the vehicle approaches obstacles, the planner struggles to generate feasible nodes that both avoid the obstacle and stay on the track.
Since the algorithm penalizes nodes near obstacles and plans a braking trajectory at each iteration, it is able to stop safely (without colliding with an obstacle) in 9 trials.
RRT has 1 crash because it cannot always generate a feasible braking trajectory.
Increasing the buffer size of the obstacles could reduce collisions, but would impact performance.
GPOPS-II is unable to plan trajectories in less than 0.5 seconds due to the number of track constraints; hence, it records 10 safe stops.

In the second experiment, the extended planning time allows RRT to generate  more nodes per planning iteration, and complete more of the track: 31\% on average, with a maximum of 86\%.
Although unable to reach the goal, RRT uses the extended planning time to find safe stopping paths.
GPOPS-II successfully reaches the goal in all 10 trials; however, it achieves an average planning time of 3.71 s.
The planning times for GPOPS-II have a standard deviation of 4.20 s; the large standard deviation is expected because the number of constraints vary based on the track curvature.
Heuristics may reduce the amount of constraints, but would be obstacle- or track-specific.
In contrast, the average planning time and standard deviation of RTD is 0.09 s and 0.06 s; hence we expect changes in the track will not affect its ability to perform in real time.
Additionally due to Lemma \ref{lem:T_sense_pers_feas}, when RTD is unable to plan a trajectory within the 0.5 s time limit, it is always able to safely brake.
\section{Conclusion}\label{sec:conclusion}

To design trajectories for autonomous cars while ensuring safety and persistent feasibility, one must have real-time performance despite model uncertainty and error in the vehicle's ability to track a planned trajectory.
In this work, we apply the Reachability-based Trajectory Design (RTD) method, which is provably safe and persistently feasible, to a full-sized passenger vehicle in CarSim.
RTD has been applied to small mobile robots in prior work; here, we demonstrate that the method can plan dynamically-feasible, safe trajectories for autonomous cars.
In ten simulated trials, RTD successfully drives the vehicle around an entire 1 km test track at up to 15 m/s around randomly generated obstacles (only known when detected at runtime) safely and in real-time.
Currently, RTD is limited to static obstacles, and requires that braking is implicitly included in the offline reachability computation.
Future work will address these limitations, apply RTD on a physical car, explore new online algorithms for globally-optimal planning, and address vehicle-specific types of uncertainty such as road friction.

\renewcommand{\bibfont}{\normalfont\small}
{\renewcommand{\markboth}[2]{}
\printbibliography}

\end{document}